\documentclass[preprint,amsmath]{revtex4}
\usepackage{amssymb}
\usepackage{graphicx,epsfig,subfigure,dsfont,amssymb,amsthm,amsfonts,amsbsy,mathrsfs,amscd,appendix}
\def\qed{\leavevmode\unskip\penalty9999 \hbox{}\nobreak\hfill
	\quad\hbox{\leavevmode  \hbox to.77778em{%
			\hfil\vrule   \vbox to.675em%
			{\hrule width.6em\vfil\hrule}\vrule\hfil}}
	\par\vskip3pt}
\usepackage{tikz}  
\usepackage{pgf}   
\usepackage{multirow}
\usepackage{graphicx}
\usepackage{subfigure}
\newtheorem{theorem}{Theorem}
\newtheorem{corollary}{Corollary}
\newtheorem{lemma}{Lemma}
\newtheorem{example}{Example}

\usepackage{amsmath}
\usepackage{enumerate}
\usepackage{amssymb}
\usepackage{graphicx}
\usepackage{algorithm}
\usepackage{algorithmic}
\usepackage{array}
\usepackage{longtable}
\graphicspath{{figures/},{pics/}}
\makeatletter

\newcommand{\Rmnum}[1]{\expandafter\@slowromancap\romannumeral #1@}
\makeatother

\begin{document}

\begin{center}
		\bf{Trade-off relations of geometric coherence}
\end{center}

\begin{center}
Bingyu Hu and Ming-Jing Zhao$^\ast$

\vspace{2ex}

{ School of Science, Beijing Information Science and Technology University, Beijing, 100192, P. R. China\\}

\end{center}

{\bf Abstract}
Quantum coherence is an important quantum resource and it is intimately related to various research fields. The geometric coherence is a coherence measure both operationally and geometrically. We study the trade-off relation of geometric coherence in qubit systems.
We first derive an upper bound for the geometric coherence by the purity of quantum states. Based on this, a complementarity relation between the quantum coherence and the mixedness is established. We then derive the quantum uncertainty relations of the geometric coherence on two and three general measurement bases in terms of the incompatibility respectively, which turn out to be state-independent for pure states.
These trade-off relations provide the limit to the amount of quantum coherence. As a byproduct,
the complementarity relation between the minimum error probability for
discriminating a pure-states ensemble and the mixedness of quantum states is established.


{$^\ast$}Correspondence: zhaomingjingde@126.com

\section{Introduction}

Coherence stemmed from superposition principle is the foundation of multiparticle interference and
plays an important role in physics \cite{A. Streltsov-rev, M. Hu,K.D. Wu}.
After the first resource framework for quantum coherence was introduced in Ref. \cite{T. Baumgratz}, the characterization and quantification of quantum coherence have attracted much attention. Some coherence measures are proposed to quantify the coherence such as
the $l_1$ norm of coherence \cite{T. Baumgratz}, the relative entropy of coherence \cite{T. Baumgratz}, the intrinsic randomness of coherence \cite{X. Yuan}, coherence concurrence \cite{X. Qi}, the robustness of coherence \cite{C. Napoli}, and the fidelity-based coherence \cite{T. Baumgratz,Liu2017} and so on.

The geometric coherence is also a coherence measure based on the fidelity \cite{Streltsov}. It amounts to the maximum bipartite geometric entanglement that can be
generated via incoherent operations applied to the system and an incoherent ancilla \cite{Streltsov}. Moreover, the geometric coherence is equal to the minimum error probability to discriminate a set of pure states with von Neumann measurement \cite{Chunhe Xiong}.
These geometric characterizations and operational interpretations make the geometric coherence prominent as the coherence measure. However, the geometric coherence is rather difficult to calculate.
So some evaluations have been provided \cite{Hai-Jun Zhang}, which are tight for some special quantum states.
The complementarity relation between the geometric coherence and the mixedness has been also revealed which characterizes further the geometric coherence quantitatively\cite{Uttam Singh,Hai-Jun Zhang,X.Che}.

More than that, the quantum coherence is basis-dependent and changes with different bases, so it is interesting and meaningful to investigate the constraints of the quantum coherence on different measurement bases. That is called the quantum uncertainty relation in the sense that coherence and quantum uncertainty are dual on the same quantum substrate \cite{Luo2017,Xiao Yuan}. Specifically, in qubit systems, the quantum uncertainty relations in terms of the relative entropy of coherence, the coherence of formation and the $l_{1}$ norm of coherence under any two orthonormal bases are derived \cite{Xiao Yuan}. In high dimensional systems,
the quantum uncertainty relations of Tsallis relative entropy of coherence and R\'{e}nyi relative entropy of coherence are derived \cite{Rastegin2023,FuGang Zhang},
which are then generalized to the sandwiched R\'{e}nyi relative entropy of coherence and unified ($\alpha$,$\beta$)-relative entropy of coherence
\cite{Haijin Mu}. The quantum uncertainty relations for coherence measure based on the Wigner-Yanase skew information are also
established \cite{Shun-Long Luo}. More generally, the quantum uncertainty relation with respect to positive operator-valued measurements is considered \cite{Rastegin2021} and the quantum uncertainty relation in bipartite systems is also discussed \cite{Singh2016}.

In this work, we investigate the trade-off relations of the geometric coherence in qubit systems.
Some basic concepts of the geometric coherence are presented in Sec. II. In Sec. III, we
derive an upper bound for the geometric coherence in terms of the purity of quantum states. Based on this, a complementarity relation between the quantum coherence and the mixedness is established. In Sec. IV, we derive the quantum uncertainty relations of the geometric coherence on two and three general measurement bases in terms of the purity of quantum states and the incompatibility of the bases.
In Sec. V, we establish a complementarity relation between the minimum error probability for discriminating a pure-state ensemble and the mixedness of quantum states. We conclude
in Sec. VI.

\section{Preliminaries}

The resource theory framework of quantum coherence is composed of two basic elements:
incoherent states and incoherent operations. Considering a $d$-dimensional Hilbert space $\mathcal{H}_{d}$ with an orthonormal
basis $\mathbb{X}=\{|x_i\rangle\}_{i=1}^d$, we call the diagonal quantum state $\rho=\sum_{i = 1}^d \lambda_i|x_i\rangle\langle x_i|$
as the incoherent state. This set of incoherent
states is labeled by $\mathcal{I}_{\mathbb{X}}$.
Quantum operation
$\Phi(\rho)=\sum_{l} K_l\rho K_l^{\dag}$ with $\sum_{l} K_l^{\dag} K_l=I$ is an incoherent operation if it fulfills
$K_l \sigma K_l^{\dag}/tr(K_l \sigma K_l^{\dag})\in \mathcal{I}_{\mathbb{X}}$ for all
$\sigma\in\mathcal{I}_{\mathbb{X}}$ and for all $l$.

The quantum coherence is degreed by a nonnegative function named the coherence measure, which should satisfy the following requirements \cite{T. Baumgratz}:
\begin{enumerate}[(C1)]
\item Nonnegativity: $C(\rho)\geq0$ and $C(\rho)=0$ for all $\rho\in\mathcal{I}_{\mathbb{X}}$;
\item Monotonicity: $C$ does not increase under the action of incoherent operations: $C(\Phi(\rho))\leq C(\rho)$ for any incoherent operation $\Phi$;
\item Strong monotonicity: $C$ does not increase on average under selective incoherent operations: $\sum_{l}p_{l}C(\rho_{l})\leq C(\rho)$ where $p_{l}=tr(K_{l}\rho K_{l}^{\dagger})$, $\rho_{l}=K_{l}\rho K_{l}^{\dagger}/p_{l}$, and $\Phi(\cdot)=\sum_{l} K_l(\cdot) K_l^{\dag}$ is an incoherent operation;
\item Convexity: $C(\rho)\leq \sum_{i} p_{i}C(\rho_{i})$ for any $\rho=\sum_{i}p_{i}\rho_{i}$ and $p_i\geq 0$, $\sum_i p_i =1$.
\end{enumerate}

The geometric coherence is a coherence measure defined as  \cite{Streltsov}
\begin{align}
	C_{g}^{\mathbb{X}}(\rho)=1-\max_{\sigma\in\mathcal{I}_{\mathbb{X}}}F(\rho,\sigma),
	\end{align}
where the maximum is taken over all incoherent states $\sigma\in\mathcal{I}_{\mathbb{X}}$, and
\begin{eqnarray}\label{eq fidelity}
F(\rho,\sigma)=(tr\sqrt{\sqrt{\rho} \sigma \sqrt{\rho}})^{2}
\end{eqnarray}
is the fidelity of states $\rho$ and $\sigma$ \cite{Jozsa}.
Here one should note that the fidelity in Eq. (\ref{eq fidelity}) is the square of the Uhlmann fidelity in Ref. \cite{Uhlmann,M.A.}.
For pure states, the geometric coherence is $C_{g}^{\mathbb{X}}(|\psi\rangle)=1-\max_{|x_{i}\rangle\in \mathbb{X}}|\langle x_{i}|\psi \rangle|^{2}$, where the
maximum is taken over all $|x_{i}\rangle\in \mathbb{X}=\{|x_i\rangle\}_{i=1}^d$. For high dimensional mixed states,
the geometric coherence is formidably difficult to calculate. Then Ref. \cite{Hai-Jun Zhang} proved the upper and lower bounds of the geometric coherence by the sub-fidelity and super-fidelity \cite{J.A. Miszczak},
which induce an analytical formula for qubit states as
\begin{align}\label{AA}
	C_{g}^{\mathbb{X}}(\rho)=\frac{1}{2}-\frac{1}{2}\sqrt{1-2(tr(\rho^{2})-\sum\limits_{i=1}\limits^2 \langle x_{i}|\rho|x_{i} \rangle^{2})}.
	\end{align}
Next we shall employ Eq. (\ref{AA}) and discuss the complementarity relation and quantum uncertainty relation of the geometric coherence in qubit systems.

\section{The complementarity relation of the geometric coherence}

In this section we shall discuss the complementarity relation of the geometric coherence in qubit systems. Throughout the paper,
we denote $\mathcal{P}=tr(\rho^{2})$ as the purity of quantum state $\rho$. The purity satisfies $\frac{1}{2}\leq \mathcal{P} \leq 1$. $\mathcal{P} = 1$ if and only if $\rho$ is pure. $\mathcal{P} = \frac{1}{2}$ if and only if $\rho$ is maximally mixed, $\rho=I/2$.
We also refer to $\mathbb{B}_{\rho}=\{|\phi_1\rangle, |\phi_2\rangle\}$ as the basis composed by the orthonormal eigenvectors of the nonmaximally mixed state $\rho$.
Then the geometric coherence $C_{g}^{\mathbb{X}}$ and the purity $\mathcal{P}$ are factually linked together closely.

\begin{theorem}\label{th1*}
For any qubit state $\rho$ and any orthonormal basis $\mathbb{X}=\{|x_{i}\rangle,i=1,2\}$, the geometric coherence $ C_{g}^{\mathbb{X}}(\rho)$ is bounded from above by the purity $\mathcal{P}$ as
\begin{align}\label{upp}
	 C_{g}^{\mathbb{X}}(\rho) \leq \frac{1}{2}-\frac{1}{\sqrt{2}}\sqrt{1-\mathcal{P}}.
	\end{align}
The inequality becomes an equality if and only if $\langle x_{1}|\rho|x_{1} \rangle=\langle x_{2}|\rho|x_{2} \rangle=\frac{1}{2}$.
	\end{theorem}

\begin{proof}
Since $2(\langle x_{1}|\rho|x_{1} \rangle ^{2}+\langle x_{2}|\rho|x_{2} \rangle ^{2})\geq (\langle x_{1}|\rho|x_{1} \rangle+ \langle x_{2}|\rho|x_{2} \rangle )^2=1$, so by Eq. \eqref{AA}, $	C_{g}^{\mathbb{X}}(\rho)=\frac{1}{2}-\frac{1}{2}\sqrt{1-2(\mathcal{P}-\langle x_{1}|\rho|x_{1} \rangle ^{2}-\langle x_{2}|\rho|x_{2} \rangle ^{2})} \leq \frac{1}{2}-\frac{1}{\sqrt{2}}\sqrt{1-\mathcal{P}}$. The inequality becomes an equality if and only if $\langle x_{1}|\rho|x_{1} \rangle=\langle x_{2}|\rho|x_{2} \rangle=\frac{1}{2}$. This completes the proof.
\end{proof}

Theorem \ref{th1*} tells us that
$ C_{g}^{\mathbb{X}}(\rho) = \frac{1}{2}-\frac{1}{\sqrt{2}}\sqrt{1-\mathcal{P}}$ if and only if
$\langle x_{1}|\rho|x_{1} \rangle=\langle x_{2}|\rho|x_{2} \rangle=\frac{1}{2}$,
which means the density matrix of $\rho$ in the orthonormal basis $\mathbb{X}$ has the equal diagonal entries $1/2$. This class of quantum states are called the  assisted maximally coherent states in Ref. \cite{MZhao} and can be decomposed as the convex combination of maximally coherent states $\frac{1}{\sqrt{2}}(|x_1\rangle + e^{{\rm i}\theta} |x_2\rangle)$. This also means the geometric coherence $ C_{g}^{\mathbb{X}}(\rho)$ reaches the upper bound $\frac{1}{2}-\frac{1}{\sqrt{2}}\sqrt{1-\mathcal{P}}$ if and only if $\rho$ is the maximally mixed state or for nonmaximally mixed states, the orthonormal basis $\mathbb{B}_{\rho}=\{|\phi_1\rangle, |\phi_2\rangle\}$ and the orthonormal basis $\mathbb{X}=\{|x_1\rangle, |x_2\rangle\}$ are mutually unbiased, $|\langle \phi_i |x_j\rangle|^2=1/2$, for $|\phi_i\rangle \in \mathbb{B}_{\rho}$ and $|x_j\rangle\in \mathbb{X}$, $\forall i, j$.

Furthermore, it is worthy to emphasize that the upper bound of geometric coherence in Theorem \ref{th1*} is basis-independent. It just depends on the purity of quantum states. This is a remarkable difference between other basis-dependent upper bounds. Now we make a comparison between our upper bound in Theorem \ref{th1*} and the upper bounds derived in Ref. \cite{Chunhe Xiong} in evaluating the geometric coherence.

\begin{example}
For maximally coherent mixed states in a qubit system \cite{Uttam Singh}
\begin{align}\label{exstate1}
	\rho_{m}=\frac{1-q}{2}I+ q|\psi^+\rangle\langle\psi^+|,
	\end{align}
where $q\in[0,1]$, $|\psi^+\rangle=\frac{1}{\sqrt{2}}(|0\rangle+|1\rangle)$ is the maximally coherent state in the computational basis $\{|0\rangle,|1\rangle\}$. If we fix $\mathbb{X}=\{|0\rangle,|1\rangle\}$, we find our upper bound in Eq. (\ref{upp}) is exactly the geometric coherence
\begin{eqnarray}\label{eq1 example1}
    C_{g}^{\mathbb{X}}(\rho_{m})=\frac{1-\sqrt{1-q^{2}}}{2}.
	\end{eqnarray}
The upper bounds derived in Ref. \cite{Chunhe Xiong} are $l_{3}=\frac{q^{2}}{2}$ and $l_{4}=q$.  Obviously,
both $l_3< C_{g}^{\mathbb{X}}(\rho_{m})$ and $l_4< C_{g}^{\mathbb{X}}(\rho_{m})$, so our upper bound has an advantage over
the upper bounds in Ref. \cite{Chunhe Xiong} for this example.
\end{example}

Additionally, Theorem \ref{th1*} also indicates a complementarity relation between the geometric coherence and the mixedness by utilizing the normalized linear entropy $S_{L}(\rho)=\frac{d}{d-1}(1-\mathcal{P})$ as the mixedness quantifier \cite{Nicholas A}.

\begin{corollary}\label{th1}
For any qubit state $\rho$ and any orthonormal basis $\mathbb{X}=\{|x_{i}\rangle,i=1,2\}$, the geometric coherence $C_{g}^{\mathbb{X}}(\rho)$ and the mixedness $S_{L}(\rho)$ should satisfy
\begin{align}
	 C_{g}^{\mathbb{X}}(\rho)+\frac{1}{2}\sqrt{S_{L}(\rho)} \leq \frac{1}{2}.
	\end{align}
The inequality becomes an equality if and only if $\langle x_{1}|\rho|x_{1} \rangle=\langle x_{2}|\rho|x_{2} \rangle=\frac{1}{2}$.
	\end{corollary}
This complementarity relation imposes a fundamental limit on the amount of coherence that can be extracted from quantum states with equal mixedness.
For any given mixedness $S_{L}(\rho)$ of quantum states ($S_{L}(\rho) \neq 1$, i.e. $\rho\neq I/2$) and given reference basis $\mathbb{X}$, the quantum states such that the orthonormal basis $\mathbb{B}_{\rho}$ and the reference basis $\mathbb{X}$ are mutually unbiased can reach the maximal possible geometric coherence.

\section{The quantum uncertainty relation of geometric coherence}

In this section, we present the quantum uncertainty relations of geometric coherence on two and three general measurement bases respectively. These quantum uncertainty relations depend on the purity of quantum state and incompatibility of bases.
For any two orthonormal bases $\mathbb{X}=\{|x_{i}\rangle,i=1,2\}$ and $\mathbb{Y}=\{|y_{i}\rangle,i=1,2\}$ in qubit systems, the incompatibility \cite{Maassen} between $\mathbb{X}$ and $\mathbb{Y}$ is
\begin{eqnarray}
\mathfrak{c}=\max_{|x_{i}\rangle\in\mathbb{X},|y_{j}\rangle\in\mathbb{Y}}|\langle x_{i}|y_{j}\rangle |^{2},
\end{eqnarray}
which is the maximum overlap of the two measurements. It describes the commutativity between two orthonormal bases. The incompatibility satisfies $\frac{1}{2} \leq \mathfrak{c} \leq 1 $.
Two orthonormal bases $\mathbb{X}=\{|x_{i}\rangle,i=1,2\}$ and $\mathbb{Y}=\{|y_{i}\rangle,i=1,2\}$ are called compatible if $\mathfrak{c}=1$, that is, $\{|y_{1}\rangle,|y_{2}\rangle\}=\{e^{{\rm i}\theta_1} |x_{1}\rangle, e^{{\rm i}\theta_2} |x_{2}\rangle\}$. Moreover, we also need the following lemma proved in Ref. \cite{Xiao Yuan} for the characterization of the feasible domain of the optimization problem next.
\begin{lemma}\label{lemma1}
Let $|x\rangle$, $|z\rangle$, $|r\rangle$ be any three $d$-dimensional normalized vectors, $a=|\langle x|r \rangle |^2$, $b=|\langle z|r \rangle |^2$, $c=|\langle z|x \rangle |^2$.
Then $a$, $b$, $c$ satisfy
	\begin{eqnarray*}
		a+b&\le 1+\sqrt{c} \ \ \text{and}\ \
|a-b|&\le \sqrt{1-c}.
	\end{eqnarray*}
	Especially, when $d=2$, we have one more inequality that
	\begin{eqnarray*}
		1-\sqrt{c}\le a+b.
	\end{eqnarray*}
\end{lemma}

\subsection{The quantum uncertainty relation of geometric coherence on two general measurement bases}

Based on the incompatibility of two orthonormal bases, the quantum uncertainty relation of the geometric coherence on two general measurement bases can be derived.

\begin{theorem}\label{th2}
For any qubit state $\rho$ and two orthonormal bases $\mathbb{X}=\{|x_{i}\rangle,i=1,2\}$ and $\mathbb{Y}=\{|y_{i}\rangle,i=1,2\}$, the geometric coherence of $\rho$ satisfies the following quantum uncertainty relation:
\begin{align}
	 C_{g}^{\mathbb{X}}(\rho)+C_{g}^{\mathbb{Y}}(\rho) \geq \frac{1}{2}[1-\sqrt{1+4(2\mathcal{P}-1)(\mathfrak{c}-\sqrt{\mathfrak{c}})}],\label{DD}
	\end{align}
where $\mathcal{P}$ is the purity of the quantum state $\rho$, and $\mathfrak{c}$ is the incompatibility of $\mathbb{X}$ and $\mathbb{Y}$.
The equality holds true if and only if $\rho$ is the maximally mixed state or for nonmaximally mixed states, the three bases $\mathbb{X}$, $\mathbb{Y}$ and ${\mathbb{B}}_{\rho}$ are all compatible.
\end{theorem}

\begin{proof}
Without loss of generality, we suppose $\mathfrak{c}= |\langle x_{1}|y_{1} \rangle|^{2}$. For any given qubit state $\rho$ with spectral decomposition $\rho=p|\phi_1\rangle \langle\phi_1|+(1-p)|\phi_2\rangle \langle \phi_2|$, we denote $a=|\langle \phi_1|x_{1} \rangle|^{2}$ and $b=|\langle \phi_1|y_{1} \rangle|^{2}$, using Eq. \eqref{AA}, we have
\begin{eqnarray*}
		C_{g}^{\mathbb{X}}(\rho)+C_{g}^{\mathbb{Y}}(\rho)
		=1-\frac{1}{2}[\sqrt{1+4(2\mathcal{P}-1)(a^{2}-a)}+\sqrt{1+4(2\mathcal{P}-1)(b^{2}-b)}].
	\end{eqnarray*}
Let $f(x)=\sqrt{1+4(2\mathcal{P}-1)(x^{2}-x)}$, $x\in[0,1]$, it is easy to verify that $f(x)$ is convex and symmetric about $x=\frac{1}{2}$. Furthermore $f$ decreases when $x\leq\frac{1}{2}$ and increases when $x\geq \frac{1}{2}$.
Now we focus on the maximum of the function $g(a,b)=f(a)+f(b)$, which is symmetric about $a=\frac{1}{2}$ and $b=\frac{1}{2}$. Here we only need to consider the case $a,b\in[0,\frac{1}{2}]$.
Hence we get the following optimization problem:
\begin{equation}\label{opt1}
\begin{aligned}
&\max \quad g(a,b)=f(a)+f(b)\\
&\begin{array}{r@{\quad}r@{}l@{\quad}l}
s.t. &&a+b\in[1-\sqrt{\mathfrak{c}},1+\sqrt{\mathfrak{c}}]\\
     &&b-a\in[-\sqrt{1-\mathfrak{c}},\sqrt{1-\mathfrak{c}}]\\
     &&a,b\in[0,\frac{1}{2}]\\
\end{array}
\end{aligned}
   \end{equation}
by Lemma 1.
The feasible domain of the optimization problem is depicted in Fig. 1.
\begin{figure}[h]\scriptsize 	
\begin{center}
		\subfigure[]
		{
				\includegraphics[width=0.6\textwidth]{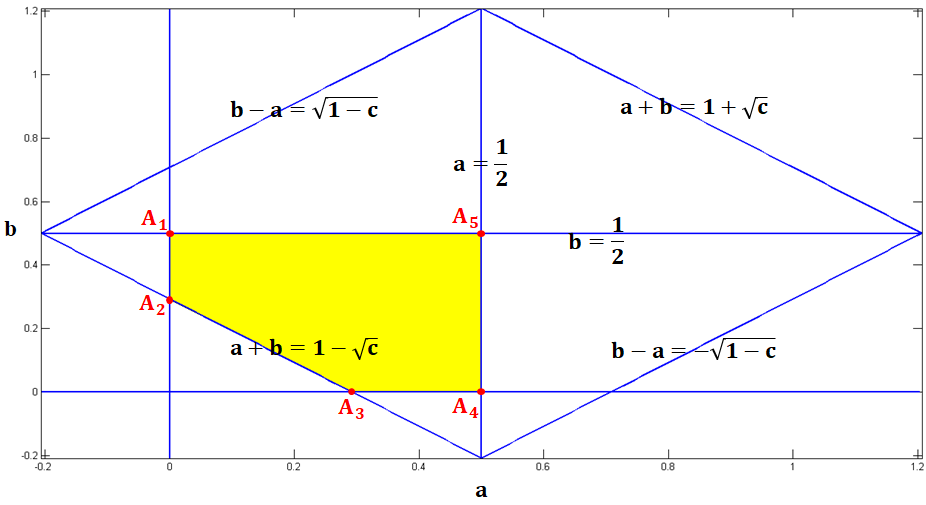}
		}
		\subfigure[]
		{
				\includegraphics[width=0.6\textwidth]{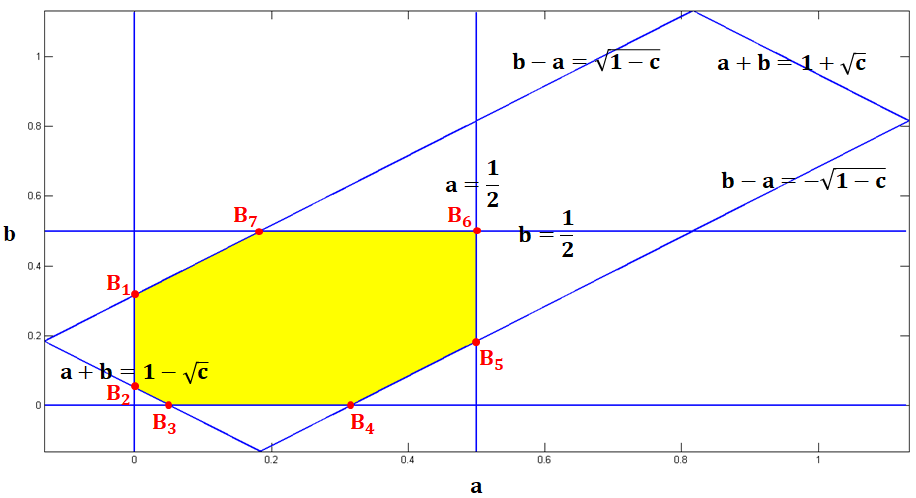}
		}
        \caption{The yellow regions represent the feasible domain of the optimization problem in Eq. \eqref{opt1}. When $\mathfrak{c}$ belongs to $[0.5,0.75]$, the feasible domain is in shape of the subfigure (a) with $A_{1}(0,\frac{1}{2})$, $A_{2}(0,1-\sqrt{\mathfrak{c}})$, $A_{3}(1-\sqrt{\mathfrak{c}},0)$, $A_{4}(\frac{1}{2},0)$ and $A_{5}(\frac{1}{2},\frac{1}{2})$ as the extreme points. When $\mathfrak{c}$ belongs to $(0.75,1]$, the feasible domain is in shape of the subfigure (b) with $B_{1}(0,\sqrt{1-\mathfrak{c}})$, $B_{2}(0,1-\sqrt{\mathfrak{c}})$, $B_{3}(1-\sqrt{\mathfrak{c}},0)$, $B_{4}(\sqrt{1-\mathfrak{c}},0)$, $B_{5}(\frac{1}{2},\frac{1}{2}-\sqrt{1-\mathfrak{c}})$, $B_{6}(\frac{1}{2},\frac{1}{2})$ and $B_{7}(\frac{1}{2}-\sqrt{1-\mathfrak{c}},\frac{1}{2})$ as the extreme points. In the subfigure (a), we choose $\mathfrak{c}=0.5$ as an example. In the subfigure (b), we choose $\mathfrak{c}=0.9$ as an example.
        }
\end{center}
\end{figure}
Because $g$ is a convex function of $a,b$ and the feasible domain is a convex polygon, the maximum of function $g$ must be obtained at the extreme points of the feasible domain.

When $\mathfrak{c}$ belongs to $[0.5,0.75]$, the extreme points of the feasible domain are
$A_{1}(0,\frac{1}{2})$, $A_{2}(0,1-\sqrt{\mathfrak{c}})$, $A_{3}(1-\sqrt{\mathfrak{c}},0)$, $A_{4}(\frac{1}{2},0)$ and $A_{5}(\frac{1}{2},\frac{1}{2})$
in the subfigure (a) of Fig. 1. 
Therefore we have
\begin{eqnarray*}
\max g(a,b)=\max\{f(0)+f(1-\sqrt{\mathfrak{c}}),f(0)+f(\frac{1}{2}),f(\frac{1}{2})+f(\frac{1}{2})\}.
\end{eqnarray*}
Since $0\leq1-\sqrt{\mathfrak{c}}<\frac{1}{2}$, so the maximum of $g$ is definitely $f(0)+f(1-\sqrt{\mathfrak{c}})$.

When $\mathfrak{c}$ belongs to $(0.75,1]$, the extreme points of the feasible domain are $B_{1}(0,\sqrt{1-\mathfrak{c}})$, $B_{2}(0,1-\sqrt{\mathfrak{c}})$, $B_{3}(1-\sqrt{\mathfrak{c}},0)$, $B_{4}(\sqrt{1-\mathfrak{c}},0)$, $B_{5}(\frac{1}{2},\frac{1}{2}-\sqrt{1-\mathfrak{c}})$, $B_{6}(\frac{1}{2},\frac{1}{2})$ and $B_{7}(\frac{1}{2}-\sqrt{1-\mathfrak{c}},\frac{1}{2})$
in the subfigure (b) of Fig. 1.
Therefore we have
\begin{eqnarray*}
\max g(a,b)=\max\{f(0)+f(1-\sqrt{\mathfrak{c}}),f(0)+f(\sqrt{1-\mathfrak{c}}),f(\frac{1}{2})+f(\frac{1}{2}-\sqrt{1-\mathfrak{c}}),f(\frac{1}{2})+f(\frac{1}{2})\}.
\end{eqnarray*}
Since $|1-\sqrt{\mathfrak{c}}-\frac{1}{2}|\geq|\sqrt{1-\mathfrak{c}}-\frac{1}{2}|$, so the maximum of $g$ is definitely $f(0)+f(1-\sqrt{\mathfrak{\mathfrak{c}}})$.

Combining the two cases above, we derive that the geometric coherence of $\rho$ on two general measurement bases $\mathbb{X}$ and $\mathbb{Y}$ should satisfy
\begin{eqnarray*}
	C_{g}^{\mathbb{X}}(\rho)+C_{g}^{\mathbb{Y}}(\rho)&=&1-\frac{1}{2}[f(a)+f(b)]\\
    &\geq&1-\frac{1}{2}[f(0)+f(1-\sqrt{\mathfrak{c}})]\\
    &=&\frac{1}{2}[1-\sqrt{1+4(2\mathcal{P}-1)(\mathfrak{c}-\sqrt{\mathfrak{c}})}].
	\end{eqnarray*}

Now we consider the saturation of the quantum uncertainty relation. First, if $\mathcal{P}=\frac{1}{2}$, then the lower bound of the geometric coherence $C_{g}^{\mathbb{X}}(\rho)+C_{g}^{\mathbb{Y}}(\rho)$ as well as the geometric coherence $C_{g}^{\mathbb{X}}(\rho)$ and $C_{g}^{\mathbb{Y}}(\rho)$ vanishes. So Eq. \eqref{DD} becomes an equality for the maximally mixed state. For nonmaximally mixed state, by the analysis above we know the geometric coherence $C_{g}^{\mathbb{X}}(\rho)+C_{g}^{\mathbb{Y}}(\rho)$ reaches its lower bound when (1) $a=1$, $b=\sqrt{\mathfrak{c}}$; (2) $a=0$, $b=1-\sqrt{\mathfrak{c}}$; (3) $a=1-\sqrt{\mathfrak{c}}$, $b=0$; (4) $a=\sqrt{\mathfrak{c}}$, $b=1$. Thanks to the symmetry of function $f$, we only need to consider the case $a=1$, $b=\sqrt{\mathfrak{c}}$. The condition $a=1$ implies  $\{|\phi_1\rangle, |\phi_2\rangle\}=\{e^{i\alpha_{1}}|x_{1}\rangle, e^{i\alpha_{2}}|x_{2}\rangle\}$. The condition $b=\sqrt{\mathfrak{c}}$ associated with $\mathfrak{c}=|\langle x_{1}|y_{1} \rangle|^{2}=|\langle \phi_1|y_{1} \rangle|^{2}=\sqrt{\mathfrak{c}}$ involves $\mathfrak{c}=1$,
which means $\{|y_{1}\rangle, |y_{2}\rangle\}=\{e^{i\beta_{1}}|x_{1}\rangle, e^{i\beta_{2}}|x_{2}\rangle\}$.
Therefore three bases $\mathbb{X}$, $\mathbb{Y}$ and ${\mathbb{B}}_{\rho}$ are all compatible. This completes the proof.
\end{proof}

Theorem \ref{th2} provides a quantum uncertainty relation for the geometric coherence with two measurement bases, in terms of the purity of quantum states and the incompatibility of the two bases. The equality in \eqref{DD} holds true only for maximally mixed states or three compatible bases $\mathbb{X}$, $\mathbb{Y}$ and ${\mathbb{B}}_{\rho}$ for nonmaximally mixed states. This is equivalent to that the quantum state $\rho$ is incoherent simultaneously under the bases $\mathbb{X}$ and $\mathbb{Y}$. Conversely, for any nonmaximally mixed states and two incompatible measurement bases $\mathbb{X}$ and  $\mathbb{Y}$, the lower bound in Eq. \eqref{DD} is always strictly positive. Particularly,
if the purity is specified to the maximum, $\mathcal{P}=1$, the quantum uncertainty relation is just relied on the incompatibility of the two bases. This demonstrates the quantum uncertainty relation is state-independent for
any pure state $|\psi\rangle$, that is,
\begin{align}
	 C_{g}^{\mathbb{X}}(|\psi\rangle)+C_{g}^{\mathbb{Y}}(|\psi\rangle) \geq \frac{1}{2}[1-\sqrt{1+4(\mathfrak{c}-\sqrt{\mathfrak{c}})}].
	\end{align}

Now we consider an example to illustrate the evaluation of the geometric coherence.
\begin{example}\label{example2}
For the maximally coherent mixed states in Eq. (\ref{exstate1}) in qubit systems, $C_{g}^{\mathbb{X}}(\rho_{m})+C_{g}^{\mathbb{Y}}(\rho_{m})$ can be evaluated from above by Eq. \eqref{upp} and from below by Eq. \eqref{DD} as
\begin{align}
	 \frac{1}{2}[1-\sqrt{1+4q^{2}(\mathfrak{c}-\sqrt{\mathfrak{c}})}] \leq C_{g}^{\mathbb{X}}(\rho_{m})+C_{g}^{\mathbb{Y}}(\rho_{m}) \leq 1-\sqrt{1-q^{2}}.\label{ex2eq1}
	\end{align}
From this equation we can see the upper bound is independent with the orthonormal bases while the lower bound increases when the incompatibility $\mathfrak{c}$ decreases.

First we fix $\mathbb{X}=\mathbb{B}_{\rho_{m}}=\{\frac{1}{\sqrt{2}}(|0\rangle+|1\rangle),\frac{1}{\sqrt{2}}(|0\rangle-|1\rangle)\}$ composed by the eigenvectors of $\rho_{m}$ and $\mathbb{Y}=\{\frac{1}{\sqrt{5}}(|0\rangle+2|1\rangle),\frac{1}{\sqrt{5}}(-2|0\rangle+|1\rangle)\}$ with the incompatibility $\mathfrak{c}=\frac{9}{10}$. Then $C_{g}^{\mathbb{X}}(\rho_{m})+C_{g}^{\mathbb{Y}}(\rho_{m})$ can be evaluated by Eq. \eqref{ex2eq1} as
\begin{align}
	 \frac{1}{2}[1-\sqrt{1+\frac{6(3-\sqrt{10})q^{2}}{5}}] \leq C_{g}^{\mathbb{X}}(\rho_{m})+C_{g}^{\mathbb{Y}}(\rho_{m}) \leq 1-\sqrt{1-q^{2}}.\label{EE}
	\end{align}
In fact by Eq. \eqref{AA} we have
\begin{eqnarray}\label{ex21}
C_{g}^{\mathbb{X}}(\rho_{m})+C_{g}^{\mathbb{Y}}(\rho_{m}) =C_{g}^{\mathbb{Y}}(\rho_{m})=\frac{1}{2}(1-\sqrt{1-\frac{9}{25}q^{2}}).
\end{eqnarray}
The geometric coherence $C_{g}^{\mathbb{X}}(\rho_{m})+C_{g}^{\mathbb{Y}}(\rho_{m})$ and the upper and lower bounds in Eq. \eqref{EE} are plotted in the first subfigure in Fig. 2.

Second, we consider $\mathbb{X}=\{\frac{1}{\sqrt{2}}(|0\rangle+i|1\rangle),\frac{1}{\sqrt{2}}(|0\rangle-i|1\rangle)\}$ and keep the basis $\mathbb{Y}=\{\frac{1}{\sqrt{5}}(|0\rangle+2|1\rangle),\frac{1}{\sqrt{5}}(-2|0\rangle+|1\rangle)\}$ with the incompatibility $\mathfrak{c}=\frac{1}{2}$.
Then $C_{g}^{\mathbb{X}}(\rho_{m})+C_{g}^{\mathbb{Y}}(\rho_{m})$ can be evaluated by Eq. \eqref{ex2eq1} as
\begin{align}
	 \frac{1}{2}[1-\sqrt{1+2(1-\sqrt{2})q^{2}}] \leq C_{g}^{\mathbb{X}}(\rho_{m})+C_{g}^{\mathbb{Y}}(\rho_{m}) \leq 1-\sqrt{1-q^{2}},\label{FF}
	\end{align}
In this case, by Eq. \eqref{AA} we have
\begin{eqnarray}\label{ex22}
C_{g}^{\mathbb{X}}(\rho_{m})+C_{g}^{\mathbb{Y}}(\rho_{m}) =\frac{1}{2}(1-\sqrt{1-q^{2}})+\frac{1}{2}(1-\sqrt{1-\frac{9}{25}q^{2}}).
\end{eqnarray}
The geometric coherence $C_{g}^{\mathbb{X}}(\rho_{m})+C_{g}^{\mathbb{Y}}(\rho_{m})$ and the upper and lower bounds in Eq. \eqref{FF} are plotted in the second subfigure in Fig. 2.

Since the incompatibility of the first set of orthonormal bases is greater than that of the second, so the lower bound of the geometric coherence under the second set of orthonormal bases is greater. This makes the region between the upper bound and lower bound in subfigure (b) in Fig. 2 smaller. So the evaluation of the geometric coherence $C_{g}^{\mathbb{X}}(\rho_{m})+C_{g}^{\mathbb{Y}}(\rho_{m})$ is more precise with smaller incompatibility.


\begin{figure}[htbp]
\centering

\subfigure[]
{
\begin{minipage}[t]{0.45\textwidth}
\includegraphics[width=0.9\textwidth,height=0.7\textwidth]{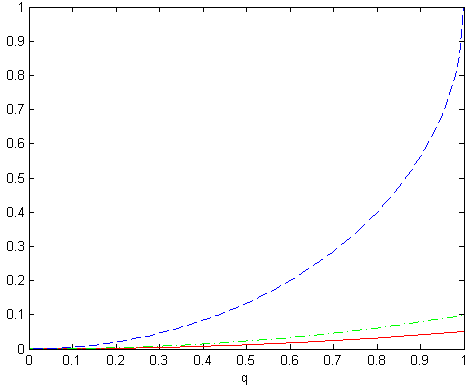}
\end{minipage}
}
\quad
\subfigure[]
{
\begin{minipage}[t]{0.45\textwidth}
\includegraphics[width=0.9\textwidth,height=0.7\textwidth]{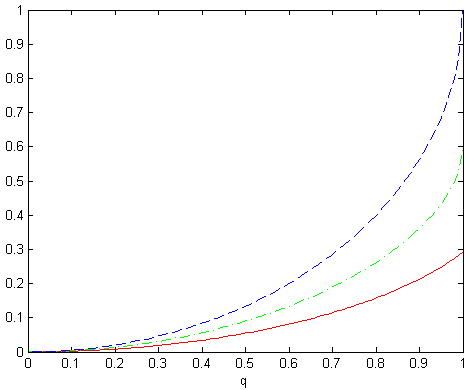}
\end{minipage}
}
\caption{The evaluation of the geometric coherence $ C_{g}^{\mathbb{X}}(\rho_m)+C_{g}^{\mathbb{Y}}(\rho_m)$. In the subfigure (a), we choose $\mathbb{X}=\{\frac{1}{\sqrt{2}}(|0\rangle+|1\rangle),\frac{1}{\sqrt{2}}(|0\rangle-|1\rangle)\}$ and $\mathbb{Y}=\{\frac{1}{\sqrt{5}}(|0\rangle+2|1\rangle),\frac{1}{\sqrt{5}}(-2|0\rangle+|1\rangle)\}$.
The green dot-dashed line is $C_{g}^{\mathbb{X}}(\rho_{m})+C_{g}^{\mathbb{Y}}(\rho_{m})$ in Eq. (\ref{ex21}).
The red solid line and blue dashed line are the lower and upper bounds in Eq. \eqref{EE}.  In the subfigure (b), we choose  $\mathbb{X}=\{\frac{1}{\sqrt{2}}(|0\rangle+i|1\rangle),\frac{1}{\sqrt{2}}(|0\rangle-i|1\rangle)\}$ and $\mathbb{Y}=\{\frac{1}{\sqrt{5}}(|0\rangle+2|1\rangle),\frac{1}{\sqrt{5}}(-2|0\rangle+|1\rangle)\}$. The green dot-dashed line is $C_{g}^{\mathbb{X}}(\rho_{m})+C_{g}^{\mathbb{Y}}(\rho_{m})$ in Eq. (\ref{ex22}).
The red solid line and blue dashed line are the lower and upper bounds in Eq. \eqref{FF}.}
\end{figure}
\end{example}

\subsection{The quantum uncertainty relation of geometric coherence with three measurement bases}

For any three orthonormal bases $\mathbb{X}=\{|x_{i}\rangle,i=1,2\}$, $\mathbb{Y}=\{|y_{i}\rangle,i=1,2\}$ and $\mathbb{Z}=\{|z_{i}\rangle,i=1,2\}$,
we define the incompatibility vector of $\mathbb{X}$, $\mathbb{Y}$, $\mathbb{Z}$ as $\vec{\mathfrak{c}}=({\mathfrak{c}}_1, \mathfrak{c}_2, \mathfrak{c}_3)$,
\begin{eqnarray}
\vec{\mathfrak{c}}=(\mathfrak{c}_1, \mathfrak{c}_2, \mathfrak{c}_3):=(\max_{|x_{i}\rangle\in\mathbb{X},|y_{j}\rangle\in\mathbb{Y}}|\langle x_{i}|y_{j}\rangle |^{2}, \max_{|y_{i}\rangle\in\mathbb{Y},|z_{j}\rangle\in\mathbb{Z}}|\langle y_{i}|z_{j}\rangle |^{2},\max_{|x_{i}\rangle\in\mathbb{X},|z_{j}\rangle\in\mathbb{Z}}|\langle x_{i}|z_{j}\rangle |^{2})^{\downarrow},
\end{eqnarray}
where the superscript $^{\downarrow}$ means rearranging the entries in descending order.
The incompatibility vector $\vec{\mathfrak{c}}=(\mathfrak{c}_1, \mathfrak{c}_2, \mathfrak{c}_3)$ describes the commutativity of three orthnormal bases. It satisfies $\frac{1}{2} \leq \mathfrak{c}_i \leq 1$ for $i=1,2,3$. $\vec{\mathfrak{c}}=(1, 1, 1)$ if and only if three orthonormal bases are all compatible.
Based on the incompatibility vector $\vec{\mathfrak{c}}$, we derive a quantum uncertainty relation of geometric coherence under three measurement bases.

\begin{theorem}\label{th3}
For any qubit state $\rho$
and three orthonormal bases $\mathbb{X}=\{|x_{i}\rangle,i=1,2\}$, $\mathbb{Y}=\{|y_{i}\rangle,i=1,2\}$ and $\mathbb{Z}=\{|z_{i}\rangle,i=1,2\}$, suppose $\vec{\mathfrak{c}}=(\mathfrak{c}_1, \mathfrak{c}_2, \mathfrak{c}_3)$ is the incompatibility vector of $\mathbb{X}$, $\mathbb{Y}$ and $\mathbb{Z}$.
\begin{enumerate}[(1)]
\item
If $1+\sqrt{\mathfrak{c}_3}<\sqrt{\mathfrak{c}_{1}}+ \sqrt{\mathfrak{c}_2}$, then
\begin{align}
	 C_{g}^{\mathbb{X}}(\rho)+C_{g}^{\mathbb{Y}}(\rho)+C_{g}^{\mathbb{Z}}(\rho) \geq 1-\frac{1}{2}[f(\sqrt{\mathfrak{c}_{1}})+f(\sqrt{\mathfrak{c}_{1}}-\sqrt{\mathfrak{c}_3})].\label{eq1 theorem3}
	\end{align}
\item
If $1+\sqrt{\mathfrak{c}_3} \geq \sqrt{\mathfrak{c}_{1}} + \sqrt{\mathfrak{c}_2}$, then
\begin{eqnarray*}
&&C_{g}^{\mathbb{X}}(\rho)+C_{g}^{\mathbb{Y}}(\rho)+C_{g}^{\mathbb{Z}}(\rho) \geq \frac{3}{2}-\frac{1}{2}\max\{1+f(\sqrt{\mathfrak{c}_{1}})+f(\sqrt{\mathfrak{c}_2}),\\
&&f(\frac{1-\sqrt{\mathfrak{c}_{1}}-\sqrt{\mathfrak{c}_{2}}+\sqrt{\mathfrak{c}_{3}}}{2})+f(\frac{1-\sqrt{\mathfrak{c}_{1}}+\sqrt{\mathfrak{c}_{2}}-
\sqrt{\mathfrak{c}_{3}}}{2})+f(\frac{1+\sqrt{\mathfrak{c}_{1}}-\sqrt{\mathfrak{c}_{2}}-\sqrt{\mathfrak{c}_{3}}}{2})\},
	\end{eqnarray*}
\end{enumerate}
with $f(x)=\sqrt{1+ 4(2\mathcal{P}-1)(x^2-x)}$ and $\mathcal{P}$ the purity of the quantum state $\rho$.
\end{theorem}

\begin{proof}
Since $\vec{\mathfrak{c}}=(\mathfrak{c}_1, \mathfrak{c}_2, \mathfrak{c}_3)$ is the incompatibility vector of three orthonormal bases $\mathbb{X}$, $\mathbb{Y}$ and $\mathbb{Z}$, without loss of generality, we
suppose $\mathfrak{c}_1=\max_{|x_{i}\rangle\in\mathbb{X},|y_{j}\rangle\in\mathbb{Y}}|\langle x_{i}|y_{j}\rangle |^{2}=|\langle x_{1}|y_{1}\rangle |^{2}$,  $\mathfrak{c}_2=\max_{|x_{i}\rangle\in\mathbb{X},|z_{j}\rangle\in\mathbb{Z}}|\langle x_{i}|z_{j}\rangle |^{2}=|\langle x_{1}|z_{1}\rangle |^{2}$, and $\mathfrak{c}_3=\max_{|y_{i}\rangle\in\mathbb{Y},|z_{j}\rangle\in\mathbb{Z}}|\langle y_{i}|z_{j}\rangle |^{2}$, and $\mathfrak{c}_1 \geq \mathfrak{c}_2 \geq \mathfrak{c}_3$.
Suppose $\rho=p|\phi_1\rangle \langle\phi_1|+(1-p)|\phi_2\rangle \langle\phi_2|$ is the spectral decomposition of $\rho$, denote $a=|\langle \phi_1|x_{1} \rangle|^{2}$, $b=|\langle \phi_1|y_{1} \rangle|^{2}$ and $n=|\langle \phi_1|z_{1} \rangle|^{2}$, then using Eq. \eqref{AA}, we have
\begin{eqnarray*}
		&&C_{g}^{\mathbb{X}}(\rho)+C_{g}^{\mathbb{Y}}(\rho)+C_{g}^{\mathbb{Z}}(\rho)\\
		&=&\frac{3}{2}-\frac{1}{2}[\sqrt{1+4(2\mathcal{P}-1)(a^{2}-a)}+\sqrt{1+4(2\mathcal{P}-1)(b^{2}-b)}+\sqrt{1+4(2\mathcal{P}-1)(n^{2}-n)}]\\
        &=&\frac{3}{2}-\frac{1}{2}[f(a)+f(b)+f(n)].
	\end{eqnarray*}
Now we consider the maximum of $G(a,b,n)=f(a)+f(b)+f(n)$.
By the symmetry of the function $G$ and Lemma 1, we only need to consider
the following optimization problem:
\begin{equation}\label{opt2}
\begin{aligned}
&\max \quad G(a,b,n)=f(a)+f(b)+f(n)\\
&\begin{array}{r@{\quad}r@{}l@{\quad}l}
s.t. &&a+b\in[1-\sqrt{\mathfrak{c_{1}}},1+\sqrt{\mathfrak{c_{1}}}]\\
     &&a+n\in[1-\sqrt{\mathfrak{c_{2}}},1+\sqrt{\mathfrak{c_{2}}}]\\
     &&b+n\in[1-\sqrt{\mathfrak{c_{3}}},1+\sqrt{\mathfrak{c_{3}}}]\\
     &&b-a\in[-\sqrt{1-\mathfrak{c_{1}}},\sqrt{1-\mathfrak{c_{1}}}]\\
     &&n-a\in[-\sqrt{1-\mathfrak{c_{2}}},\sqrt{1-\mathfrak{c_{2}}}]\\
     &&a,b,n\in[0,\frac{1}{2}],
\end{array}
\end{aligned}
   \end{equation}
where we have extended the range of $b+n$ from $[1-\sqrt{|\langle y_{1}|z_{1} \rangle|^{2}},1+\sqrt{|\langle y_{1}|z_{1} \rangle|^{2}}]$ to $[1-\sqrt{\mathfrak{c_{3}}},1+\sqrt{\mathfrak{c_{3}}}]$ and remove the restriction of $n-b$ for simplicity.

First if $\mathfrak{c}_{1}\in[0.5, 0.75]$, then we divide the discussion of the optimization problem in Eq. (\ref{opt2}) into two cases.

(1) When $1+\sqrt{\mathfrak{c_3}}< \sqrt{\mathfrak{c_{1}}} + \sqrt{\mathfrak{c_2}} $, the feasible domain is a convex polyhedron in the first subfigure in Fig. 3. Since $G(a,b,n)$ is a convex function of $a,b,n$, its maximum over this feasible domain must be achieved at the extreme points, which are the
vertices of the polyhedron $C_{1}(1-\sqrt{\mathfrak{c_{1}}},0,1-\sqrt{\mathfrak{c_{3}}})$, $C_{2}(0,1-\sqrt{\mathfrak{c_{1}}},\sqrt{\mathfrak{c_{1}}}-\sqrt{\mathfrak{c_{3}}})$, $C_{5}(0,\sqrt{\mathfrak{c_{2}}}-\sqrt{\mathfrak{c_{3}}},1-\sqrt{\mathfrak{c_{2}}})$, $C_{6}(1-\sqrt{\mathfrak{c_{2}}},1-\sqrt{\mathfrak{c_{3}}},0)$ and the vertices
\begin{equation}\label{point d}
\begin{array}{rcl}
&&D_{1}(1-\sqrt{\mathfrak{c_{1}}},0,\frac{1}{2}),\ \  D_{2}(0,1-\sqrt{\mathfrak{c_{1}}},\frac{1}{2}),\ \  D_{3}(0,\frac{1}{2},\frac{1}{2}), \ \ D_{4}(\frac{1}{2},\frac{1}{2},\frac{1}{2}),\\
&&D_{5}(\frac{1}{2},0,\frac{1}{2}), \ \ D_{6}(0,\frac{1}{2},1-\sqrt{\mathfrak{c_{2}}}), \ \ D_{7}(1-\sqrt{\mathfrak{c_{2}}},\frac{1}{2},0),\ \ D_{8}(\frac{1}{2},\frac{1}{2},0),\\
&&D_{9}(\frac{1}{2},1-\sqrt{\mathfrak{c_{3}}},0),\ \  D_{10}(\frac{1}{2},0,1-\sqrt{\mathfrak{c_{3}}}).
\end{array}
\end{equation}
By calculating and comparing these values, we get $\max G(a,b,n)=\max\{G(C_{2}),G(C_{5})\}$. Let $H(x)=f(0)+f(x)+f(1-\sqrt{\mathfrak{c_{3}}}-x)$, $x\in[0,1-\sqrt{\mathfrak{c_{3}}}]$, one can verify that $H(x)$ is a convex function and symmetric about $x=\frac{1-\sqrt{\mathfrak{c_{3}}}}{2}$. Since $|1-\sqrt{\mathfrak{c_{1}}}-\frac{1-\sqrt{\mathfrak{c_{3}}}}{2}|\geq |\sqrt{\mathfrak{c_{2}}}-\sqrt{\mathfrak{c_{3}}}-\frac{1-\sqrt{\mathfrak{c_{3}}}}{2}|$, the value of $C_{5}$ is no more than the value of  $C_{2}$, therefore the maximum of function $G$ is definitely
\begin{eqnarray}
\max G(a,b,n)=G(0,1-\sqrt{\mathfrak{c}_{1}},\sqrt{\mathfrak{c}_{1}}-\sqrt{\mathfrak{c}_3}).
\end{eqnarray}

(2) When $1+\sqrt{\mathfrak{c}_3}\geq \sqrt{\mathfrak{c}_{1}} + \sqrt{\mathfrak{c}_2} $, the feasible domain is a convex polyhedron in the second subfigure in Fig. 3. Similarly, to get the maximum of function $G$ in this polyhedron, we need to calculate the values of
the vertices $C_{1}(1-\sqrt{\mathfrak{c_{1}}},0,1-\sqrt{\mathfrak{c_{3}}})$, $C_{3}(0,1-\sqrt{\mathfrak{c_{1}}},1-\sqrt{\mathfrak{c_{2}}})$, $C_{6}(1-\sqrt{\mathfrak{c_{2}}},1-\sqrt{\mathfrak{c_{3}}},0)$, $C_{7}(\frac{1-\sqrt{\mathfrak{c_{1}}}-\sqrt{\mathfrak{c_{2}}}+\sqrt{\mathfrak{c_{3}}}}{2},\frac{1-\sqrt{\mathfrak{c_{1}}}+\sqrt{\mathfrak{c_{2}}}-\sqrt{\mathfrak{c_{3}}}}{2},
\frac{1+\sqrt{\mathfrak{c_{1}}}-\sqrt{\mathfrak{c_{2}}}-\sqrt{\mathfrak{c_{3}}}}{2})$ and the vertices $\{D_{i},i=1,...,10\}$ in Eq. (\ref{point d}). By comparing the values, we get
\begin{eqnarray}
\max G(a,b,n)=\max \{G(0,1-\sqrt{\mathfrak{c_{1}}},1-\sqrt{\mathfrak{c_2}}),G(\frac{1-\sqrt{\mathfrak{c_{1}}}-\sqrt{\mathfrak{c_{2}}}+\sqrt{\mathfrak{c_{3}}}}{2}, \nonumber\\ \frac{1-\sqrt{\mathfrak{c_{1}}}+\sqrt{\mathfrak{c_{2}}}-\sqrt{\mathfrak{c_{3}}}}{2}, \frac{1+\sqrt{\mathfrak{c_{1}}}-\sqrt{\mathfrak{c_{2}}}-\sqrt{\mathfrak{c_{3}}}}{2})\}.
\end{eqnarray}

\begin{figure}
\begin{center}
        \subfigure[]
		{
				\includegraphics[width=0.49\textwidth]{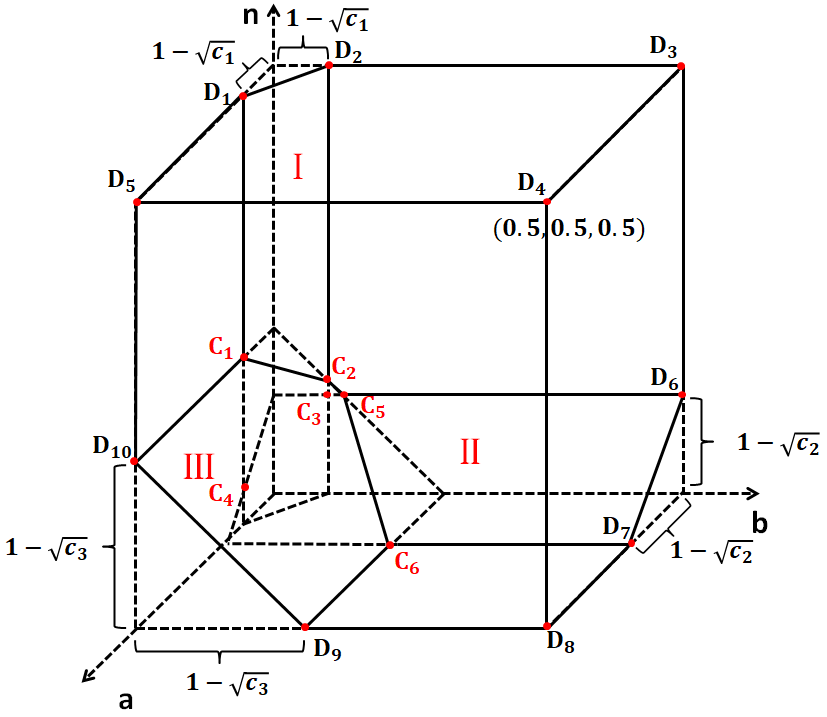}
		}
		\subfigure[]
		{
				\includegraphics[width=0.49\textwidth]{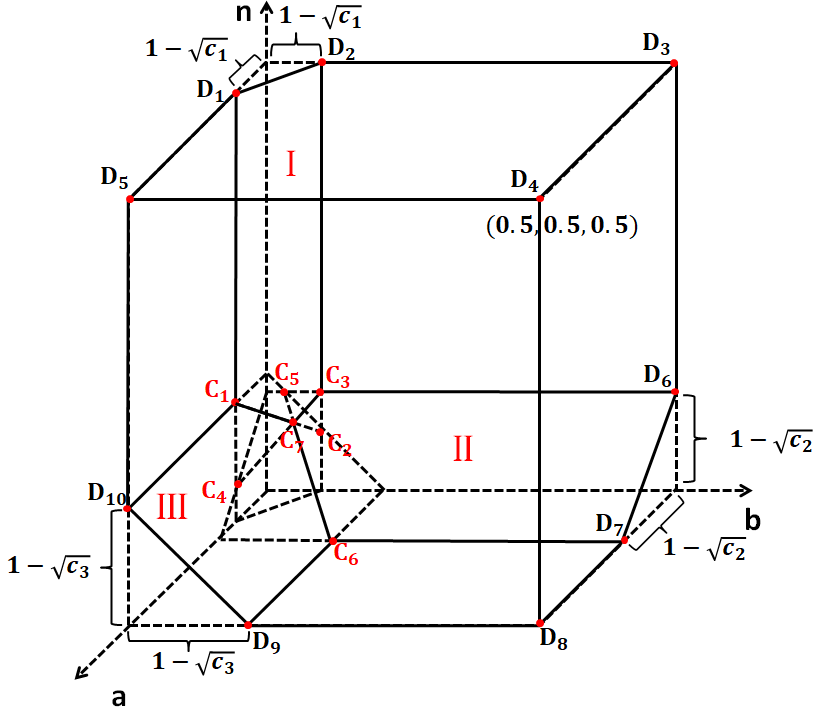}
		}
		\caption{The geometry of the feasible domain in the optimization problem in Eq. \eqref{opt2} for $\mathfrak{c}_{1}\in [0.5, 0.75]$. When $1+\sqrt{\mathfrak{c}_3}< \sqrt{\mathfrak{c}_{1}} + \sqrt{\mathfrak{c}_2} $, the feasible domain is in shape of the subfigure (a). When $1+\sqrt{\mathfrak{c}_3}\geq \sqrt{\mathfrak{c}_{1}} + \sqrt{\mathfrak{c}_2}$, the feasible domain is in shape of the subfigure (b). Both polyhedra are formed by cutting three triangular prisms from a cube with length $\frac{1}{2}$. We label these cross-sectional planes as \Rmnum{1}, \Rmnum{2} and \Rmnum{3} respectively. The intersection points of the cross-sectional planes of \Rmnum{1} and \Rmnum{3} are $C_{1}$ and $C_{2}$, the intersection points of the cross-sectional planes of \Rmnum{1} and \Rmnum{2} are $C_{3}$ and $C_{4}$, the intersection points of the cross-sectional planes of \Rmnum{2} and \Rmnum{3} are $C_{5}$ and $C_{6}$. The vertex $C_{7}$ in subfigure (b) is the intersection point of \Rmnum{1}, \Rmnum{2} and \Rmnum{3}. $D_{i}$ ($1 \leq i \leq 10$) are other vertices of the polyhedra.}	
\end{center}
\end{figure}

Second, if $\mathfrak{c}_{1}\in (0.75,1]$, then the feasible domain in the optimization problem in Eq. (\ref{opt2}) is a smaller polyhedron contained in Fig. 3. So the maximum reached in the first case is an upper bound for this case, which completes the proof.
\end{proof}

Theorem \ref{th3} provides a quantum uncertainty relation for the geometric coherence under three measurement bases, in terms of the purity of quantum states and the incompatibility vector of three bases. Similarly, this quantum uncertainty relation is also state-independent for
pure state. Next we illustrate the lower bound of the geometric coherence by an explicit example.

\begin{example}\label{example3}
For the maximally coherent mixed states in Eq. (\ref{exstate1}) in qubit systems, we fix $\mathbb{X}=\{\frac{1}{\sqrt{5}}(-2|0\rangle+|1\rangle),\frac{1}{\sqrt{5}}(|0\rangle+2|1\rangle)\}$, $ \mathbb{Y}=\{\frac{1}{\sqrt{2}}(|0\rangle-|1\rangle),\frac{1}{\sqrt{2}}(|0\rangle+|1\rangle)\}$ and $\mathbb{Z}=\{|0\rangle,|1\rangle\}$. By calculation we get the incompatibility vector of $\mathbb{X}$, $\mathbb{Y}$, $\mathbb{Z}$ is $\vec{\mathfrak{c}}=(\mathfrak{c}_1, \mathfrak{c}_2, \mathfrak{c}_3)=(\frac{9}{10},\frac{4}{5},\frac{1}{2})$. Since $1+\sqrt{\mathfrak{c}_3}< \sqrt{\mathfrak{c}_{1}} + \sqrt{\mathfrak{c}_2}$, by Eq. \eqref{eq1 theorem3} we have the lower bound of $C_{g}^{\mathbb{X}}(\rho_{m})+C_{g}^{\mathbb{Y}}(\rho_{m})+C_{g}^{\mathbb{Z}}(\rho_{m})$ is
\begin{align}
C_{g}^{\mathbb{X}}(\rho_{m})+C_{g}^{\mathbb{Y}}(\rho_{m})+C_{g}^{\mathbb{Z}}(\rho_{m})\geq1-\frac{1}{2}\left[\sqrt{1+\frac{6(3-\sqrt{10})q^{2}}{5}}+\right.\nonumber\\
\left.\sqrt{1+\frac{2(14-6\sqrt{5}-3\sqrt{10}+5\sqrt{2})q^{2}}{5}}\,\right]. \label{eq1 exemple3}
	\end{align}
In fact the geometric coherence is
\begin{align}	 C_{g}^{\mathbb{X}}(\rho_{m})+C_{g}^{\mathbb{Y}}(\rho_{m})+C_{g}^{\mathbb{Z}}(\rho_{m})=\frac{1}{2}(1-\sqrt{1-\frac{9}{25}q^{2}})+\frac{1}{2}(1-\sqrt{1-q^{2}}),\label{eq2 exemple3}
	\end{align}
by Eq. (\ref{AA}). The geometric coherence $C_{g}^{\mathbb{X}}(\rho_{m})+C_{g}^{\mathbb{Y}}(\rho_{m})+C_{g}^{\mathbb{Z}}(\rho_{m})$ and
its lower bound in Eq. \eqref{eq1 exemple3} are plotted in Fig. 4.
\begin{figure}
		\centering
		\includegraphics[width=9cm]{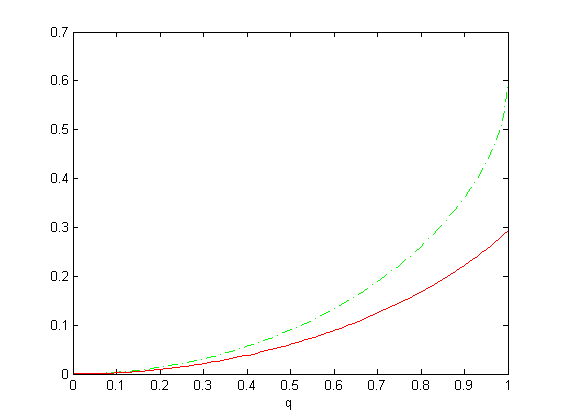}
		\caption{The evaluation of the geometric coherence $C_{g}^{\mathbb{X}}(\rho_{m})+C_{g}^{\mathbb{Y}}(\rho_{m})+C_{g}^{\mathbb{Z}}(\rho_{m})$. The green dot-dashed line is exact value in Eq. \eqref{eq2 exemple3} and the red solid line is the lower bound in Eq. \eqref{eq1 exemple3}.}	
	\end{figure}
\end{example}

\section{The geometric coherence and quantum state discrimination}

In this section we show the application of the geometric coherence in quantum state discrimination. First, let us recall the process of the ambiguous quantum state discrimination with von Neumann measurement. Suppose the sender draws at some states from $\{\rho_i\}$ with
probabilities $\{p_i\}$ and sends them to the receiver. Then the receiver tries to determine which state he has
received by a von Neumann measurement $\{\Pi_i\}$ satisfying $\sum_i \Pi_i =I$.
He declares the state is $\rho_j$ if the measurement result is $j$.
Since the probability to get the result $j$ is
$p(j|i) = tr(\Pi_j \rho_i)$ provided the state is $\rho_i$, so the maximal success probability to
discriminate this set of states $\{p_i,\rho_i\}$ is $\max_{\Pi_i} \sum_i p_i tr(\Pi_i \rho_i)$, where the maximum is over all von Neumann measurement $\{\Pi_i\}$. Correspondingly, the minimum error probability is
\begin{eqnarray}
P_e(\{p_i,\rho_i\})=1-\max_{\Pi_i} \sum_i p_i tr(\Pi_i \rho_i).
\end{eqnarray}
Ref. \cite{Chunhe Xiong} shows the geometric coherence of a quantum state amounts to the minimum error probability to discriminate a set of pure states via von Neumann measurement.
\begin{lemma}\cite{Chunhe Xiong}\label{lemma dis}
For any given orthonormal basis $\mathbb{X}=\{ |x_i\rangle \}$ in a $d$-dimensional system and any quantum state $\rho$,
the geometric coherence of $\rho$ is equal to the minimum error probability to discriminate the pure-state ensemble $\{\eta_i, |\psi_i\rangle\}$ with von Neumann measurement
\begin{eqnarray}
C_{g}^{\mathbb{X}}(\rho)=P_e(\{\eta_i,|\psi_i\rangle\}),
\end{eqnarray}
with $\eta_i=\langle x_i|\rho |x_i\rangle$ and $|\psi_i\rangle=\eta_i^{-1/2} \sqrt{\rho} |x_i\rangle$.
\end{lemma}

In fact, for any given quantum state $\rho$, the ensemble $\{\eta_i,|\psi_i\rangle\}$ with $\eta_i=\langle x_i|\rho |x_i\rangle$ and $|\psi_i\rangle=\eta_i^{-1/2} \sqrt{\rho} |x_i\rangle$ is a pure-state ensemble of $\rho$, that is, $\rho=\sum_i \eta_i |\psi_i\rangle\langle \psi_i|$. So the minimum error probability to discriminate this pure-state ensemble
corresponds to the geometric coherence of the associated mixed state. In light of the relation between the geometric coherence and purity in Theorem \ref{th1*},
the minimum error probability to discriminate any pure-state ensemble and the purity are connected.
\begin{corollary}\label{cor4}
For any qubit systems, the minimum error probability to discriminate any pure-state ensemble $\{p_i, |\psi_i\rangle\}_{i=1}^2$ by von Neumann measurement is bounded from above as
\begin{align}\label{eq dis}
	 P_e(\{p_i, |\psi_i\rangle\}_{i=1}^2) \leq \frac{1}{2}-\frac{1}{\sqrt{2}}\sqrt{1-\mathcal{P}},
	\end{align}
with $\mathcal{P}$ the purity of the associated quantum state $\rho=\sum_{i=1}^2 p_i |\psi_i\rangle \langle\psi_i|$.
	\end{corollary}

\begin{proof}
Suppose $\{p_i, |\psi_i\rangle\}_{i=1}^2$ is any pure-state ensemble in the qubit system, and $\rho=\sum_{i=1}^2 p_i |\psi_i\rangle \langle\psi_i|$ is the quantum state associated to it. We further suppose $\rho=\sum_{j=1}^2 \lambda_j |\phi_j\rangle \langle\phi_j|$ is the spectral decomposition of $\rho$. Then the two pure-state decompositions of $\rho$ are related as $\sqrt{p_i} |\psi_i\rangle=\sum_j \sqrt{\lambda_j} U_{ij} |\phi_j\rangle$, $\forall i$, for some unitary operator $U=(U_{ij})$ \cite{M.A.}. Since the eigenvectors of $\rho$ compose an orthonormal basis in the qubit system and $U=(U_{ij})$ is unitary, so there exist another orthonormal basis denoted as $\mathbb{X}=\{ |x_i\rangle \}$ such that $U_{ij}=\langle \phi_j |x_i\rangle$. Under the selected orthonormal basis $\mathbb{X}$, we have $\sqrt{p_i} |\psi_i\rangle=\sqrt{\rho} |x_i\rangle$, $\forall i$. Therefore the pure-state ensemble $\{p_i, |\psi_i\rangle\}_{i=1}^2$ satisfies $|\psi_i\rangle=p_i^{-1/2} \sqrt{\rho} |x_i\rangle$ and $p_i=\langle x_i|\rho |x_i\rangle$. By Lemma \ref{lemma dis} and Theorem \ref{th1*}, we have $P_e(\{\eta_i,|\psi_i\rangle\})=C_{g}^{\mathbb{X}}(\rho)\leq \frac{1}{2}-\frac{1}{\sqrt{2}}\sqrt{1-\mathcal{P}}$, which completes the proof.
\end{proof}

Corollary \ref{cor4} shows the minimum error probability to discriminate any pure-state ensemble  $\{p_i, |\psi_i\rangle\}_{i=1}^2$ is up bounded by the purity of the corresponding mixed states. Different from Lemma \ref{lemma dis}, the minimum error probability is just relied on the purity and independent of the reference basis.
Eq. (\ref{eq dis}) can be rewritten in form of the complementarity relation as
\begin{align}\label{eq dis com}
	 P_e(\{p_i, |\psi_i\rangle\}) +\frac{1}{{2}}\sqrt{S_L(\rho)} \leq \frac{1}{2},
	\end{align}
which imposes a limit on the amount of the minimum error probability for any pure-state ensemble  $\{p_i, |\psi_i\rangle\}_{i=1}^2$ with given mixedness.

\begin{example}
For any orthonormal basis $\mathbb{X}=\{ |x_i\rangle \}$ in qubit system, let
\begin{eqnarray}
\begin{array}{rcl}
|\psi_1\rangle=\cos\theta|x_1\rangle + \sin\theta|x_2\rangle,\\
|\psi_2\rangle=\sin\theta|x_1\rangle + \cos\theta |x_2\rangle,
\end{array}
\end{eqnarray}
and $p_1=p_2=1/2$ with $\theta\in[-\frac{\pi}{4},\frac{\pi}{4}]$. Now we consider the discrimination of the ensemble $\{p_i, |\psi_i\rangle \}$ by von Neumann measurement. On one hand, the minimum  error probability to discriminate the ensemble $\{p_i, |\psi_i\rangle \}$ is $P_e(\{p_i, |\psi_i\rangle\})=\sin^{2}\theta$ and the optimal von Neumann measurement $\{\Pi_i\}$ is $\Pi_1=|x_{1}\rangle \langle x_{1}|$ and $\Pi_2=|x_{2}\rangle \langle x_{2}|$ by Ref. \cite{Helstrom}. On the other hand, the quantum state with respect to this ensemble $\{p_i, |\psi_i\rangle \}$ is
\begin{equation}
\rho=\frac{1}{{2}} (|x_1\rangle \langle x_1|+ |x_2\rangle \langle x_2|) + \sin\theta\cos\theta( |x_1\rangle \langle x_2|+ |x_2\rangle \langle x_1|)
\end{equation}
with purity $\mathcal{P}=\frac{1}{2}+2\sin^{2}\theta\cos^{2}\theta$. In this case, the minimum  error probability reaches the upper bound in Corollary \ref{cor4}.
\end{example}

\section{Conclusions}

To summarize, we study the trade-off relations of the geometric coherence in qubit systems.
We first show the complementarity relation between the geometric coherence and the purity of quantum states. Then we derive the quantum uncertainty relations of the geometric coherence on two and three general measurement bases respectively, in terms of the purity of the quantum states and the incompatibility of the bases. At last, we obtain the complementarity relation between the minimum error probability for discriminating a pure-states ensemble and the mixedness of quantum states. These trade-off relations not only provide the limit to the amount of quantum coherence, but also reveal the intimate relations between different quantities. Since the geometric coherence is important both geometrically and operationally, we hope this work is helpful for the characterization and application of the geometric coherence.

Generalizations of the results to general high dimensional systems and more orthonormal bases are
natural extensions of this work. On one hand, for
general $m$ orthonormal bases, the quantum uncertainty relation of geometric coherence depends on the incompatibility vector of dimensional $\frac{m(m-1)}{2}$.
Although it is not feasible to depict the feasible domain as Fig. 1 or Fig. 3, a lower bound may be derived by convex analysis. On the other hand,
for general high dimensional systems, the incompatibility of two orthonormal bases may be not enough to get a good lower bound. One may need to
make use of more parameters such as the second largest value of their inner products.

\section{Acknowledgement}
This work is supported by the National Natural Science Foundation of
China under grant Nos. 12171044.

\end{document}